\documentclass[preprint,amsmath]{revtex4}
\usepackage{amssymb,amsmath,amsthm,amsfonts,amsbsy,mathrsfs,amscd,dsfont}
\usepackage{graphicx,epsfig,subfigure,multirow,algorithm,algorithmic,array,longtable}
\usepackage{tikz,pgf,pgfplots,pgfplotstable}
\pgfplotsset{compat=1.17}
\usepackage[colorlinks,linkcolor=blue,anchorcolor=blue,citecolor=blue]{hyperref}
\usepackage{appendix}
\usepackage{enumerate}
\usepackage{amsthm}
\usepackage{graphicx}

\def\qed{\leavevmode\unskip\penalty9999 \hbox{}\nobreak\hfill
	\quad\hbox{\leavevmode  \hbox to.77778em{%
			\hfil\vrule   \vbox to.675em%
			{\hrule width.6em\vfil\hrule}\vrule\hfil}}
	\par\vskip3pt}

\newtheorem{theorem}{Theorem}
\newtheorem{corollary}{Corollary}

\newtheorem{example}{Example}

\graphicspath{{figures/},{pics/}}

\makeatletter

\newcommand{\Rmnum}[1]{\expandafter\@slowromancap\romannumeral #1@}
\makeatother

\begin{document}

\begin{center}
	\bf{Coherence measures in the strictly incoherent operation framework and its application in the multi-path interferometer}
\end{center}

\begin{center}
 Peiru Li,   Jingyan Liu, Ming-Jing Zhao$^\ast$
    
 \small  School of Science, Beijing Information Science and Technology University, Beijing, 102206, P. R. China \\

{$^\ast$}Correspondence: {zhaomingjingde@126.com} 

 \end{center}

{\bf Abstract}
Quantifying coherence is an essential endeavor in both quantum foundations and quantum technologies. In this paper, we study the coherence measures in terms of the diagonal states in the strictly incoherent operations framework. Specifically, we propose a coherence measure in terms of fidelity and provide its analytical expression. The relations between the proposed coherence measure and some other coherence measures are derived. Furthermore, we prove its monotonicity under incoherent operations. As an application, we explore the role of the proposed coherence measure in 
characterizing the waveness in the multi-path interferometer. 
As a result, some wave-particle dualities in terms of fidelity are presented.
This work not only
deepens the interpretation of the diagonal states on characterizing quantum states, but also
promotes the quantitative description of the wave-particle behaviors in the multi-path interferometer.

%{\bf Keywords} 

\section{Introduction}

Quantum coherence, arising from the superposition principle, serves as a resource enabling applications beyond the reach of classical mechanics. It supports diverse research directions such as quantum computation \cite{Nielsen2010}, quantum correlations \cite{Ma2016}, quantum biology and transport phenomena \cite{Lambert2013}. 
Moreover, quantum coherence has also been linked to work extraction \cite{Francica2020} and battery capacity \cite{MXLUO}  in quantum thermodynamics,  quantum chaos \cite{Anand2021}, and light-driven superconducting dynamics \cite{Luo2023}.

In the coherence resource theory, the framework construction starts with selecting a Hilbert space basis for the quantum system, under which diagonal states are free states. The free operations depend on distinct physical motivations \cite{Streltsov2017b}, including physically incoherent operations (PIOs) \cite{ChitambarPIO}, strictly
incoherent operations (SIOs), genuinely incoherent operations (GIOs) \cite{VicenteGIO}, incoherent operations (IOs) \cite{Streltsov2017b},  dephasing-covariant incoherent operations (DIOs), and maximal incoherent operations (MIOs). Different coherence resource frameworks can be established with respect to different free operations, while the IO framework is the most commonly discussed among them \cite{Baumgratz2014}.

In the IO framework, different coherence measures are introduced respectively from the perspectives of geometric distance  \cite{Baumgratz2014,Streltsov2015,Shao,Y. Jing},  
convex roof construction \cite{XYUAN,XQI,Liu2017},  robustness of coherence \cite{Napoli}, skew information coherence \cite{CSYU, LUO2018} and so on. For the coherence measures in the geometric approach, which is given by the minimal distance between $\rho$ and the set of incoherent states, they involve two aspects including the minimal distance and the optimal incoherent state. It shows that under the fixed reference basis $\{|i\rangle\}$, the diagonal state $\Delta(\rho)=\sum_i |i\rangle\langle i|\rho|i\rangle\langle i|$ is the optimal incoherent state reaching the minimal distance for the $l_1$ norm coherence \cite{Baumgratz2014}, the relative entropy coherence \cite{Baumgratz2014}, and the coherence measure based on the absolute norm \cite{Y. Jing}, where $\Delta$ denotes the dephasing operation. But it fails for the trace norm coherence \cite{Shao} or geometric coherence generally \cite{MJ}.
As the diagonal state $\Delta(\rho)$ is a natural incoherent state related to $\rho$, we wonder whether we can characterize the coherence by it straightforwardly. As expected, if we restrict to the SIO framework, we find that
the diagonal state $\Delta(\rho)$ is able to quantify the coherence in $\rho$ essentially and geometrically. Specifically, we propose a computable coherence measure $ C_{\mathrm{diag}}^{F}$ based on the fidelity in the SIO framework.
%and derive the analytical expression. 
%Additionally, we find $ C_{\mathrm{diag}}^{F}$ is also monotonic and acts as a well-defined coherence measure in IO framework in qubit systems. 

Generally, SIOs are a class of completely positive and trace-preserving (CPTP) maps that cannot generate or make use of coherence. 
In fact, 
the quantum operations in the interferometers are SIOs  \cite{Winter2016, Yadin2016}. 
This motivates us to explore the role of coherence measures in the SIO framework in the interferometer.

In the interferometer,  when a particle passes through the paths, it acts as both particle and wave. The particle nature corresponds to the ability to obtain path information, while the wave nature is manifested in the visibility of interference patterns. 
The wave-particle duality explains quantitatively 
that the wave and particle behaviors are mutually exclusive.
The first quantitative wave-particle duality relation in multi-path interferometers was proposed by D\"{u}rr, who also established criteria for the wave and particle measures \cite{durr}. The experiments about the complementarity relations in the interferometer are also demonstrated \cite{Ding2025,KSUN}.

%Recent studies have shown that wave–particle duality is not only reflected in the complementary relationship between which-path information and interference visibility \cite{Chen2022}, and the properties of the quantum source\cite{Qian2020}.

Recently, with the development of quantum coherence, the coherence measures are considered as the wave measure and some related wave-particle duality relations have been derived \cite{MH,EB,s3,n6,Kim,Bai2025,Ding2025}. 
However, the difference between the coherence measures and wave measures is not clear up to now. 
Following the criteria of the coherence measures and wave measures, we know that the coherence measures in the IO framework are surely a wave measure. But the converse is not necessarily true. For example, the waveness measure in \cite{durr} is based on $l_2$ norm, which violates the probabilistic monotonicity under IOs. Here we extend the quantification of waveness from the coherence measures in IO framework to that in SIO framework. 
In other words, we find the coherence measures in the SIO framework are able to serve as  the wave measure in the interferometer. 
Particularly,  we show a tight wave-particle duality in the two-path interferometer and a finer wave-particle duality in the multi-path interferometer based on our proposed coherence measure.
These complementarity relations limit the behaviors of particle and wave quantitatively.

This paper is structured as follows. In Sec.  \ref{preliminaries}, we present the coherence resource theory in the SIO framework and present a method for constructing coherence measures. In Sec. \ref{coherence measure}, we present an explicit coherence measure $ C_{\mathrm{diag}}^{F}$ using fidelity and provide its analytic expression. We also establish its relationship with some other coherence measures. Furthermore, we prove its monotonicity under IOs in qubit systems. In Sec. \ref{interferometer}, we discuss the applications of the coherence measures in the SIO framework in the interferometer and derive some wave-particle dualities in terms of fidelity. Finally, Sec. \ref{Conclusions} concludes the paper.

\section{Coherence measure in the SIO framework}
\label{preliminaries}

In the coherence resource theory in the SIO framework, the free states are incoherent states and the free operations are SIOs. More specifically, in a \(d\)-dimensional Hilbert space equipped with the orthonormal basis \(\{|i\rangle\}_{i = 1}^{d}\), a quantum state $\rho$ is called incoherent if its density matrix in this basis is diagonal. The set consisting of all incoherent states is denoted as \(\mathcal{I}\). In contrast, quantum states with density matrices that are non-diagonal with respect to this basis are coherent.
%SIOs can be characterized by the Kraus representations \(\Lambda(\rho) = \displaystyle\sum_{\mu} K_{\mu} \rho K_{\mu}^{\dagger}\) \cite{Baumgratz2014,Winter2016,Yao2015,Yadin2016}, with strictly incoherent operators \(K_{\mu}\) satisfying both \(\displaystyle\sum_{\mu} K_{\mu} \mathcal{I} K_{\mu}^{\dagger} \subset \mathcal{I}\) and \(\displaystyle\sum_{\mu} K_{\mu}^{\dagger} \mathcal{I} K_{\mu} \subset \mathcal{I}\).
%This implies that there exists at most one non-zero element in each row or each column of \(K_{\mu}\) \cite{Yao2015,Du2015}. This guarantees that each strictly incoherent operator $K_{\mu}$ commutes with the dephasing map \(\Delta(\cdot) = \displaystyle\sum_{i} |i\rangle\langle i|(\cdot)|i\rangle\langle i|\). 
SIOs can be characterized by the Kraus representations
\(\Lambda(\rho) = \sum_{\mu} K_{\mu} \rho K_{\mu}^{\dagger}\)
\cite{Baumgratz2014,Winter2016,Yao2015,Yadin2016}, where each strictly incoherent
Kraus operator \(K_\mu\) satisfies both
\(K_\mu \mathcal{I} K_\mu^\dagger \subset \mathcal{I}\) and
\(K_\mu^\dagger \mathcal{I} K_\mu \subset \mathcal{I}\).
These conditions imply that \(K_\mu\) has at most one non-zero entry in each
column and at most one non-zero entry in each row \cite{Yao2015,Du2015}.
Consequently, \(K_\mu\) is compatible with the dephasing map \(\Delta(\cdot) = \sum_i |i\rangle\langle i|(\cdot)|i\rangle\langle i|\) in the sense that \(\Delta(K_\mu \rho K_\mu^\dagger)=
K_\mu \Delta(\rho) K_\mu^\dagger .\)

A nonnegative function $C$ on quantum states is called a coherence measure in the SIO framework if $C$ satisfies the following five conditions \cite{Baumgratz2014}.
(C1) Faithfulness: \(C(\rho) \geq 0 \) and \( C(\rho) = 0 \) if and only if \( \rho \) is incoherent.
(C2) Monotonicity: $C(\Lambda(\rho)) \leq C(\rho)$ if $\Lambda$ is SIO.
(C3) Probabilistic monotonicity: \(\sum_{\mu} \mathrm{Tr}( K_{\mu}\rho K_{\mu}^\dagger) C\left[ \frac{K_{\mu}\rho K_{\mu}^\dagger}{\mathrm{Tr}(K_{\mu}\rho K_{\mu}^\dagger)} \right] \leq C(\rho) \) for any SIO \( \Lambda (\cdot)=\sum_{\mu}  K_{\mu}  (\cdot) K_{\mu}^\dagger \).
(C4) Convexity: \( C(\sum_j p_j\rho_j) \leq \sum_j p_j C(\rho_j) \) for any probability distribution \( \{p_j\} \) and quantum states \( \{\rho_j\} \).
(C5) Additivity for direct sum states: \( C[p\rho_1 \oplus (1-p)\rho_2] = pC(\rho_1) + (1-p)C(\rho_2) \), where \( p \in [0,1] \), \( \rho_1, \rho_2 \) are arbitrary quantum states.
In light of the equivalence of the conditions for coherence measure under the IO framework \cite{Yu2016}, we find that the conditions (C2), (C3) and (C4) together imply (C5).  Conversely, conditions (C3) and (C4) can be derived from (C2) and (C5). So to construct a coherence measure under the SIO framework, one needs to verify that it satisfies either conditions (C1), (C2) and (C5), or conditions (C1), (C2), (C3) and (C4).

Since the set of SIOs is a proper subset of IOs, the coherence measures under IOs are also those under SIOs. 
For example, the $l_1$ norm coherence \cite{Baumgratz2014} 
 \begin{equation}
C_{l_1}(\rho) =  \|\rho - \Delta(\rho)\|_{l_1} = \sum_{i\neq j} |\langle i|\rho|j\rangle|,
\end{equation} 
is defined from the $l_1$ norm with $\|A\|_{l_1}=\sum_{i,j}|A_{ij}|$ for any matrix $A=(A_{ij})$.
The relative entropy coherence \cite{Baumgratz2014} 
 \begin{equation}
C_{r}(\rho)= S(\rho|| \Delta(\rho))=S(\Delta(\rho))-S(\rho) 
\end{equation}
is defined from the relative entropy, where $S(\rho\|\sigma) = \mathrm{Tr}\bigl[\rho(\log\rho - \log\sigma)\bigr]$ is the quantum relative entropy and $S(\rho) = -\mathrm{Tr}(\rho\log\rho)$ is the von Neumann entropy.  
The geometric coherence \cite{Streltsov2015}
\begin{equation}\label{eq cg}
C_{\mathrm{g}}(\rho) = 1 - \max_{\delta\in\mathcal{I}}F^2(\rho,\delta),
\end{equation}
is defined by the fidelity  $F(\rho,\delta) = \mathrm{Tr}\sqrt{\sqrt{\rho}\delta\sqrt{\rho}}$.  These three are all well-defined coherence measures in the IO as well as the SIO framework.

For any quantum states $\rho$ and $\sigma$, 
suppose $D$ defined on any pair of states $ (\rho, \sigma)$ is a real function  satisfying 
\begin{enumerate}[(D1)]
 \item (Nonnegativity) $D(\rho,  \Delta(\rho))\geq 0$ and $D(\rho, \Delta(\rho))= 0$ if and only if $\rho= \Delta(\rho)$.
\item (Contractivity under SIOs)  $D(\Phi(\rho), \Phi(\sigma)) \leq D(\rho,\sigma)$ for arbitrary SIO $\Phi$.
\item (Additivity under direct sum states) $D(p_1\rho_1\oplus p_2\rho_2, p_1\sigma_1\oplus p_2\sigma_2) = p_1 D(\rho_1, \sigma_1)+ p_2 D(\rho_2, \sigma_2)$ for any probability distribution $\{p_1,p_2\}$ and density operators $\rho_i$ and $\sigma_i$ supporting on the same subspace, $i=1,2$.
\end{enumerate}
Here the function $D$ is not necessarily a metric. We do not demand the function $D$ is symmetric about two components or obeys the triangle inequality.  
Based on such function $D$, for any quantum state $\rho$, we define
\begin{equation}\label{def rd}
C_{\mathrm{diag}}(\rho) = D(\rho, \Delta(\rho)).
\end{equation}

\begin{theorem}\label{th mea d}
$C_{\mathrm{diag}}(\rho)$ is a well-defined coherence measure in the  SIO framework.
\end{theorem}

\begin{proof}
First, the faithfulness (C1) is obvious.
Second, for any quantum state $\rho$ and strictly incoherent operation $\Phi$, it has
\begin{eqnarray*}
   C_{\mathrm{diag}} (\Phi(\rho))&=&D(\Phi(\rho), \Delta(\Phi(\rho)))\\
    &=& D(\Phi(\rho), \Phi(\Delta(\rho)))\\
    &\leq& D(\rho, \Delta(\rho))\\
    &=&C_{\mathrm{diag}} (\rho),
\end{eqnarray*}
%where the first equality is the definition in Eq. (\ref{def rd}), the second equality is the property of strictly incoherent operation $\Phi$, the third inequality is 
by the contractivity of the function $D$. So $C_{\mathrm{diag}}(\rho)$ satisfies the monotonicity (C2).
Third, for an arbitrary probability distribution $\{p_1,p_2\}$ and density operators $\rho_1$ and $\rho_2$,
\begin{eqnarray*}
 C_{\mathrm{diag}}(p_1\rho_1\oplus p_2\rho_2)&=&
D(p_1\rho_1\oplus p_2\rho_2, \Delta(p_1\rho_1\oplus p_2\rho_2)) \\
&=& D(p_1\rho_1\oplus p_2\rho_2, p_1 \Delta(\rho_{1}) \oplus p_2 \Delta(\rho_{2}))\ \\
&=&   p_1 D(\rho_1, \Delta(\rho_{1}))+ p_2 D(\rho_2, \Delta(\rho_{2}))\\
&=&  p_1  C_{\mathrm{diag}}(\rho_1) +p_2  C _{\mathrm{diag}}(\rho_2).
\end{eqnarray*}   
%where the first equality is the definition of $C_{\mathrm{diag}}(\rho)$ in Eq. (\ref{def rd}), the second equality is the property of direct sum, the third equality is the additivity under direct sum states of the function $D$. 
So $C_{\mathrm{diag}}(\rho)$ satisfies the condition (C3). Therefore, $C_{\mathrm{diag}}(\rho)$ is a well-defined coherence measure in the SIO framework.
\end{proof}

Theorem \ref{th mea d} presents a method for constructing coherence measures using diagonal states $\Delta(\rho)$. Since the dephasing operation $\Delta$ relies on the choice of the reference basis,  the proposed coherence measure  $C_{\mathrm{diag}}(\rho)$ is an inherently basis-dependent quantity.
A key advantage of this kind of coherence measures lies in their avoidance of the optimization process, which renders the calculation—governed by function \(D\)—more convenient and efficient. In fact, the $l_1$ norm coherence and relative entropy coherence can be viewed as this kind of coherence measures, with the function $D$ specified as the $l_1$ distance $D_1(\rho, \sigma)=||\rho-\sigma||_{l_1}$ and the relative entropy $D_2(\rho, \sigma)=S(\rho||\sigma)$ respectively.
We find a similar idea of the coherence measure has already been proposed in a general resource theory in \cite{Ziwen}, as well as the imaginarity measures in \cite{Liu2026}. Here we aim to study it systematically and specifically in the SIO framework.

\section{The fidelity-based coherence measure}
\label{coherence measure}

Now for any quantum states $\rho$ and $\sigma$, we specify the function $D$ based on the fidelity as $D(\rho, \sigma) = 1-F(\rho, \sigma)$. One can verify that this function meets three conditions (D1)-(D3). Subsequently, we derive the coherence measure in the SIO framework,
\begin{equation}
  C_{\mathrm{diag}}^{F} (\rho) = 1 - F(\rho, \Delta(\rho)).  
\end{equation}
First, in qubit systems, the coherence measure \(C_{\mathrm{diag}}^{F}(\rho)\) can be expressed by the Bloch vector of the quantum state.

\begin{theorem}   \label{th qubit eq}
For any qubit state $\rho = \frac{1}{2}(\mathbb{I} + \boldsymbol{r} \cdot \boldsymbol{\sigma})$, with $ \boldsymbol{r}=(r_1, r_2, r_3)^T$, $\boldsymbol{\sigma}=(\sigma_1, \sigma_2, \sigma_3)^T$, three Pauli operators \(\sigma_1 = |0\rangle\langle1| + |1\rangle\langle0|\), \(\sigma_2 = {\rm i}(|0\rangle\langle1| - |1\rangle\langle0|)\), \(\sigma_3 = |0\rangle\langle0| - |1\rangle\langle1|\), \(\vert \boldsymbol{r} \vert = \sqrt{r_{1}^{2} + r_{2}^{2} + r_{3}^{2}} \leq 1\), the coherence measure $C_{\mathrm{diag}}^{F}(\rho)$ is 
\begin{equation} \label{eq:diag_coherence}
        C_{\mathrm{diag}}^{F}(\rho) = 1 - \frac{\sqrt{2}}{2}  \left[ \sqrt{(1 - |\boldsymbol{r}|^2)(1 - r_3^2)} + (1 + r_3^2) \right]^{\frac{1}{2}}.
        \end{equation}
\end{theorem}

    \begin{proof} 
     For any qubit state $\rho = \frac{1}{2}(\mathbb{I} + \boldsymbol{r} \cdot \boldsymbol{\sigma})$, its diagonal part $\Delta(\rho) = \frac{1}{2}(\mathbb{I} + r_3 \sigma_3 )$, their determinants are $\det(\rho) = \frac{1 - |\mathbf{r}|^2}{4}$ and $\det(\Delta(\rho)) = \frac{1 - r_3^2}{4}$ respectively. 
        Recall that the fidelity between  qubit states $\rho$ and $\tau$ \cite{Miszczak2009,Zhang2018} is 
        \begin{eqnarray}  \label{eq qubit cdiag}
        F(\rho, \tau) = \left[ 2 \sqrt{\det(\rho) \det(\tau)} + \mathrm{Tr}(\rho \tau) \right]^{\frac{1}{2}}. 
        \end{eqnarray}
         Substituting $\label{eq:tr_rho_rho_diag}\mathrm{Tr}\left(\rho\Delta(\rho)\right) = \frac{1 + r_3^2}{2}$ into Eq. \eqref{eq qubit cdiag}, 
     we directly obtain Eq. \eqref{eq:diag_coherence}.
    \end{proof}

From Theorem \ref{th qubit eq} we see that for any given purity $|\boldsymbol{r}|$ of the quantum states, the coherence measure $C_{\mathrm{diag}}^{F}(\rho)$ decreases with $|r_3|$. This behavior is similar to the relative entropy coherence.
More generally, in higher dimensional systems, the coherence measure  $C_{\mathrm{diag}}^{F}(\rho)$ can be expressed in the form of the trace norm.

\begin{theorem}\label{th c d}
    For any $d$-dimensional quantum state $\rho$ with the spectral decomposition $\rho = \sum_{j=1}^d \lambda_j |u_j\rangle\langle u_j|$, suppose its diagonal state $\Delta(\rho) = \sum_{i=1}^d p_i |i\rangle\langle i|$ under the reference basis $\{|i\rangle\}_{i=1}^d$, then the coherence measure $C_{\mathrm{diag}}^{F}(\rho)$
is
   \begin{equation}\label{eq c d}
     C_{\mathrm{diag}}^{F}(\rho) =1- \| M \|_{\mathrm{tr}},  
   \end{equation}
    where the matrix $M=(M_{ij})_{d \times d}$ with $M_{ij} = \sqrt{\lambda_ip_j}\,\langle u_i|j\rangle$ and $\|M\|_{\mathrm{tr}}=\mathrm{Tr}\sqrt{M^{\dagger}M}$ denotes the trace norm. 
\end{theorem}

\begin{proof}
    For any $d$-dimensional quantum state $\rho$ with the spectral decomposition $\rho = \sum_{j=1}^d \lambda_j |u_j\rangle\langle u_j|$, let $|\psi\rangle = \sum_{i=1}^{d} \sqrt{\lambda_i} \, |u_i\rangle\otimes |i\rangle $
    be the purification of $\rho$. Correspondingly, 
   for the diagonal state $\Delta(\rho) = \sum_{i=1}^d p_i |i\rangle\langle i|$, let $|\phi\rangle = \sum_{i=1}^{d} \sqrt{p_{i}} \, |i\rangle\otimes |i\rangle$
      be the purification of $\Delta(\rho)$.
  According to \cite{Nielsen2010}, the fidelity between $\rho$ and  $\Delta(\rho)$ is 
   \[
    \begin{aligned}
    F(\rho, \Delta(\rho)) 
    &= \max_{U} \bigl| \langle\psi| (I \otimes U) |\phi\rangle \bigr| \\
    &= \max_{U} \bigl|\sum_{i,j} \sqrt{\lambda_i p_j} \, \langle u_i | j \rangle \, \langle i | U | j \rangle  \bigr|\\
    &=\max_{U} \bigl| \operatorname{Tr}(M U^T) \bigr|\\
    &= \| M \|_{\operatorname{tr}},
    \end{aligned}
    \] 
    where $U$ is a unitary operator on the auxiliary system,  and $M=(M_{ij})$ with $M_{ij}=\sqrt{\lambda_i p_j}\,\langle u_i|j\rangle$. The maximum is attained if we choose $U^T=V^\dagger$ for the polar decomposition of $M$ as $M=|M|V$. Therefore we obtain the formula in Eq. (\ref{eq c d}).
\end{proof}

Theorem \ref{th c d} shows the amount of coherence is determined by the eigenvalues and eigenvectors of the quantum state as well as the reference basis and the diagonal entries of the density matrix. 
However, if only the diagonal entries of the density matrix are available, the coherence measure \(C_{\mathrm{diag}}^{F}\) can be estimated from above by the following upper bound, which becomes exact for pure states.

\begin{theorem} \label{properties}
    In $d$-dimensional systems,  for any quantum state  $\rho=\displaystyle\sum_{i,j=1}^{d} \rho_{ij} |i\rangle\langle j|$, the coherence measure $C_{\mathrm{diag}}^{F}(\rho)$ is upper bounded by 
    \begin{equation}   \label{eq C pu}
        C_{\mathrm{diag}}^{F}(\rho) \leq 1 - \sqrt{
        \displaystyle\sum_{i=1}^{d}  \rho_{ii}^2},
    \end{equation}
and the equality holds if and only if $\rho$ is a pure state. 
\end{theorem}

\begin{proof}
For any mixed state $\rho$, the coherence measure $C_{\mathrm{diag}}^{F}$ is 
    \begin{equation}
    \label{upper bound on mixed states}
    \begin{split}
    C_{\mathrm{diag}}^{F}(\rho)%=1-F(\rho,\Delta(\rho)) 
    &= 1-\mathrm{Tr}\sqrt{\sqrt{\rho}\Delta(\rho)\sqrt{\rho}} \\
    & \leq 1-\sqrt{\mathrm{Tr}\left(\sqrt{\rho}\Delta(\rho)\sqrt{\rho}\right)} \\
  %  &= 1-\sqrt{\mathrm{Tr}(\rho\Delta(\rho))} \\
    &= 1-\sqrt{\sum_{i}\rho_{ii}^2}.
    \end{split}
    \end{equation}
    The first inequality arises from the positive semi-definiteness of \( \sqrt{\rho}\Delta(\rho)\sqrt{\rho} \) and the inequality becomes an equality if and only if $\rho$ is rank-1 (i.e., pure). 
    This completes the proof.
\end{proof}

From Theorem \ref{properties}, we know $ 0\leq C_{\mathrm{diag}}^{F}(\rho)\leq 1-\frac{\sqrt{d}}{d}$. The minimum is attained, $C_{\mathrm{diag}}^{F}(\rho)=0$, if and only if $\rho$ is incoherent,  $\rho= \Delta(\rho)$. The maximum is attained, $ C_{\mathrm{diag}}^{F}(\rho) = 1-\frac{\sqrt{d}}{d}$, if and only if  $\rho$ is maximally coherent, that is, $\rho=|\Psi^{\text{mcs}}\rangle\langle \Psi^{\text{mcs}}|$ with $|\Psi^{\text{mcs}}\rangle = \frac{1}{\sqrt{d}} \displaystyle\sum_{i=1}^{d} e^{{\rm i} \theta_i} |i\rangle$.
Next, we illustrate the estimation in Theorem \ref{properties} by an explicit example.

\begin{example}
Consider the \(d\)-dimensional quantum state
\begin{equation}\label{eq ex-1}
    \rho = p|0\rangle\langle 0| + (1-p)|\psi\rangle\langle \psi|,
\end{equation}
where $|\psi\rangle $ is any  pure state generated randomly and \(p \in [0, 1]\).
For the cases $d=2,5,10$, we have calculated both the exact value $ C_{\mathrm{diag}}^{F}$ and the upper bound provided by Theorem \ref{properties} (See FIG. \ref{fig:figure2}). 
The plots show that the upper bound is attained exactly when $p=0$ or $p=1$.

\begin{figure}[ht] 
    \centering
    \includegraphics[width=0.8\textwidth]{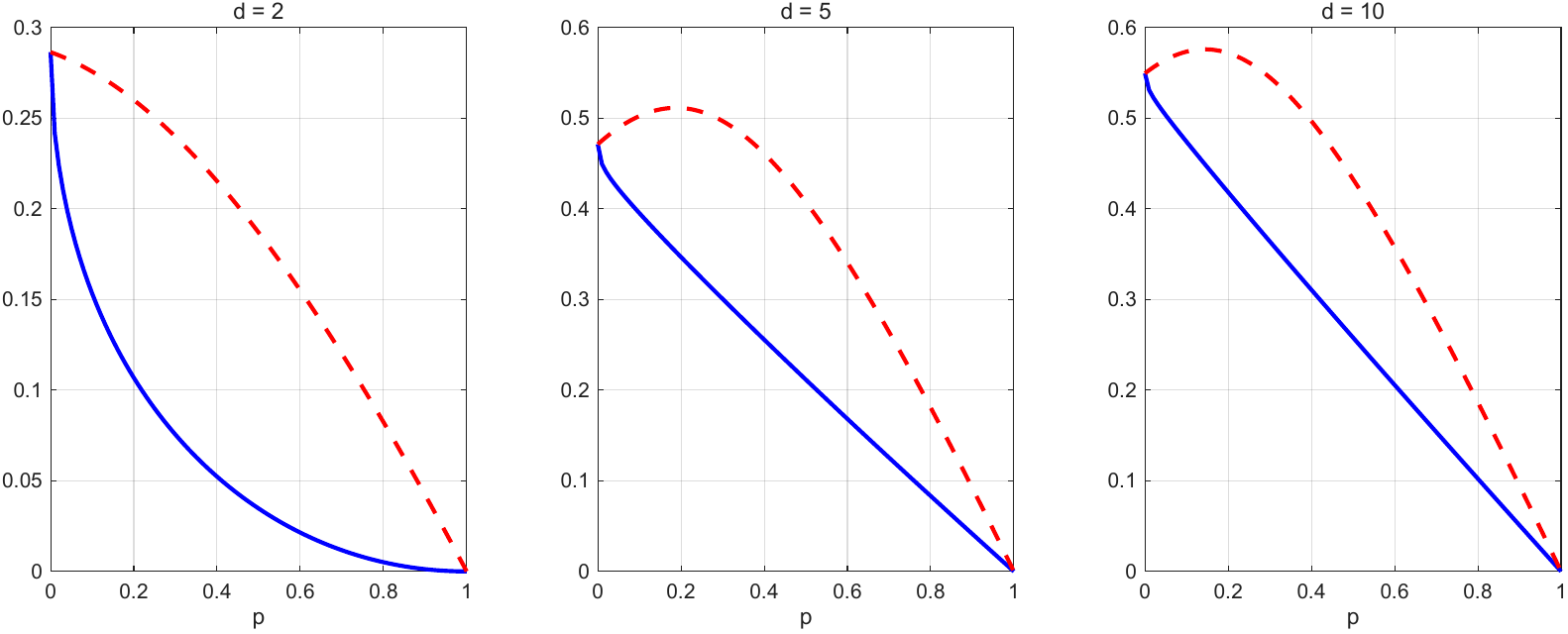}
    \caption{(Color Online) The comparison of the exact value $ C_{\mathrm{diag}}^{F}$ and its upper bound in Theorem \ref{properties} for quantum states in Eq. (\ref{eq ex-1}). The solid blue curve represents the exact value $ C_{\mathrm{diag}}^{F}$, while the dashed red curve represents the upper bound given by Theorem \ref{properties}.}
     \label{fig:figure2}
\end{figure}
\end{example}

Based on Theorem  \ref{properties}, the relationship between the coherence measure $ C_{\mathrm{diag}}^{F}(\rho)$ and geometric coherence can be established.

\begin{theorem}
\label{cg}
For any quantum state \( \rho \), the geometric coherence \( C_{\mathrm{g}}(\rho) \) and the coherence measure \( C_{\mathrm{diag}}^{F} \) satisfy 
    \begin{equation}
     1 - \sqrt{1 - C_{\mathrm{g}}(\rho)} \leq C_{\mathrm{diag}}^{F}(\rho) \leq C_{\mathrm{g}}(\rho).
    \end{equation}
\end{theorem}

\begin{proof}
The first inequality  %$1 - \sqrt{1 - C_{\mathrm{g}}(\rho)} \leq C_{\mathrm{diag}}^{F}(\rho) $ 
is obvious by the fact that
      $F(\rho,\Delta(\rho)) \leq  \max_{\delta\in\mathcal{I}}F(\rho,\delta)$.
Now we prove the second inequality.
For any pure state $|\psi\rangle=\sum_i \psi_i |i\rangle$, 
      we know 
      \begin{equation}\label{eq cg cd p}
          C_{\mathrm{diag}}^{F}(|\psi\rangle)=1- \sqrt{\displaystyle\sum_i |\psi_{i} |^4}\leq 1 - \displaystyle\max_i \psi_{i}^2=C_{\mathrm{g}}(|\psi\rangle),
      \end{equation}
           where we have used the relation \( \displaystyle\sum_i |\psi_{i} |^4 \geq \max_i |\psi_{i} |^4 \).
 As we know, the geometric coherence  \(C_{\mathrm{g}}(\rho)\) in Eq. (\ref{eq cg}) can also be expressed in the convex roof construction, that is, \(C_{\mathrm{g}}(\rho) = \min_{\{p_k, |\psi_k\rangle\}}\sum_k p_k C_{\mathrm{g}}(|\psi_k\rangle)\), where the minimum runs over all possible pure-state decompositions of $\rho$ \cite{Streltsov2000,Li2019}.
        For any mixed state \(\rho\), suppose
        \(\rho = \sum_k p_k |\psi_k\rangle\langle\psi_k|\)
        is the optimal pure-state decomposition of \(\rho\) that achieves the minimum average coherence of \(C_{\mathrm{g}}(\rho)\), namely,  \(C_{\mathrm{g}}(\rho) =\sum_k p_k C_{\mathrm{g}}(|\psi_k\rangle)\).
         Thus we have 
      \begin{equation}
          C_{\mathrm{g}}(\rho)=\sum_k p_k C_{\mathrm{g}}(|\psi_k\rangle ) \geq \sum_k p_k C_{\mathrm{diag}}^{F}(|\psi_k\rangle )\geq C_{\mathrm{diag}}^{F} (\rho),
      \end{equation}
      due to Eq. (\ref{eq cg cd p}) and the convexity of coherence measure $C_{\mathrm{diag}}^{F}$.    This completes the proof.
\end{proof}

Furthermore, Ref. \cite{{Liu2017}} proposes a fidelity-based coherence measure in convex roof construction as
\begin{equation}
C_F(\rho)=\min_{\{p_k, |\psi_k\rangle\}}\sum_k p_k C_F(|\psi_k\rangle),
\end{equation}
where the minimum runs over all possible pure-state decompositions of $\rho$, and 
\begin{equation}\label{eq cf}
    C_F(|\psi\rangle)=\min_{\sigma\in \cal{I}} [1-F^2(|\psi\rangle,\sigma)]^\frac{1}{2}.
\end{equation}
By the definition, it is obvious that the geometric coherence and the fidelity-based coherence in Eq. (\ref{eq cf}) are related by
\begin{equation}\label{eq cf cg}
    C_F^2(|\psi\rangle)=C_{\mathrm{g}} (|\psi\rangle)
\end{equation}
 for any pure state. Moreover, the fidelity-based coherence measure and \( C_{\mathrm{diag}}^{F} \) are related as follows.
 
\begin{theorem}
For any quantum state $\rho$, the coherence measure \( C_{\mathrm{diag}}^{F}(\rho) \)
and the fidelity-based coherence measure $C_F(\rho)$ satisfy
\begin{equation}\label{eq cf cdf}
1-\sqrt{1-C_F^2(\rho)}
\leq
C_{\rm diag}^F(\rho)
\leq
C_F(\rho).
\end{equation}
\end{theorem}

\begin{proof}
On one hand, let $\rho=\sum_k p_k|\psi_k\rangle\langle\psi_k|$ be an optimal pure-state decomposition of $\rho$ for the fidelity-based coherence $C_F(\rho)$, that is, $C_F(\rho)=\sum_k p_k C_F(|\psi_k\rangle)$. 
For such pure-state decomposition, 
we  have
\begin{eqnarray*}
   C_{\rm diag}^F(\rho) &\leq&     C_{\mathrm{g}}(\rho) \\
    &\leq& \sum_k p_k C_{\mathrm{g}}(|\psi_k\rangle)\\
&=&
\sum_k p_k C_F^2(|\psi_k\rangle) \\
&\leq&
\sum_k p_k C_F(|\psi_k\rangle) \\
&=&
C_F(\rho),
\end{eqnarray*}
where the first inequality is by Theorem \ref{cg}, the second inequality is by the definition of geometric coherence $C_{\mathrm{g}}$, the third line is by Eq. (\ref{eq cf cg}), the last inequality is due to $C_F^2(|\psi\rangle)
\leq C_F(|\psi\rangle)$ as $0\leq C_F(|\psi\rangle)\leq 1$. Therefore we derive the second inequality in Eq. (\ref{eq cf cdf}).

On the other hand, assume $\rho=\sum_k p_k|\psi_k\rangle\langle\psi_k|$ is an optimal pure-state decomposition of $\rho$ for the geometric coherence $C_{\mathrm{g}}(\rho)$, that is, $C_{\mathrm{g}}=\sum_k p_k C_{\mathrm{g}}(|\psi_k\rangle)$. 
For such pure-state decomposition, we have
\begin{eqnarray}\label{eq cf cg rho}
    C_F^2(\rho)
\leq \sum_k p_k C_F^2(|\psi_k\rangle)=\sum_k p_k C_{\mathrm{g}}(|\psi_k\rangle)=  C_{\mathrm{g}}(\rho),
\end{eqnarray}
where we have used the convexity of the square function for the first inequality, and the middle equality is by Eq. (\ref{eq cf cg}). 
By virtue of Eq. (\ref{eq cf cg rho}), we have $1-\sqrt{1-C_F^2(\rho)}\leq 1-\sqrt{1-C_{\mathrm{g}}(\rho)}$. Together with Theorem \ref{cg}, we obtain $1-\sqrt{1-C_F^2(\rho)}\leq C_{\rm diag}^F(\rho)$. Therefore we derive the first inequality in Eq. (\ref{eq cf cdf}).
This completes the proof.
\end{proof}

Additionally, we know these three coherence measures we just discussed satisfy the relation  $C_{\mathrm{diag}}^{F}(\rho) \leq C_F(\rho)\leq  C_{\mathrm{g}}(\rho)$ for any quantum state $\rho$ by Eq. (\ref{eq cf cg rho}).
In fact, based on Theorem \ref{th qubit eq}, the relationship between the coherence measure $C_{\mathrm{diag}}^{F}(\rho)$ and  the \( l_1 \) norm coherence can also be established.
  
\begin{theorem}
\label{q bound}
   For any qubit state $\rho$, the coherence measure $C_{\mathrm{diag}}^F (\rho)$ and the \( l_1 \) norm coherence \( C_{l_1}(\rho) \) satisfy 
   \begin{equation}
    1 - \frac{\sqrt{2}}{2}\left[2 - \frac{C_{l_1}^2(\rho)}{2}\right]^{\frac{1}{2}} \leq C_{\mathrm{diag}}^{F}(\rho) \leq 1 - \frac{\sqrt{2}}{2}\left[2 - C_{l_1}^2(\rho)\right]^{\frac{1}{2}},
    \end{equation}
    where the first equality holds if and only if $\rho $ is incoherent and  the second equality holds if and only if $\rho $ is pure. 
\end{theorem}

\begin{proof}
  For any qubit state $\rho = \frac{1}{2}(\mathbb{I} + \boldsymbol{r} \cdot \boldsymbol{\sigma})$ with $ \boldsymbol{r}=(r_1, r_2, r_3)^T$,  \(\vert \boldsymbol{r} \vert = \sqrt{r_{1}^{2} + r_{2}^{2} + r_{3}^{2}} \leq 1\), by calculation we get its  \( l_1 \) norm coherence is $C_{l_1}(\rho) = \sqrt{r_1^2 + r_2^2}$ \cite{Baumgratz2014}. On one hand, we have
        \begin{align*}
       & \sqrt{(1 - \vert \boldsymbol{r} \vert^2)(1 - r_3^2)} + (1 + r_3^2)\\
        &\leq \frac{(1 - \vert \boldsymbol{r} \vert^2) + (1 - r_3^2)}{2} + (1 + r_3^2) 
        \\& = 2 - \frac{r_1^2 + r_2^2}{2},
        \end{align*}
        where the  inequality follows from the arithmetic-geometric mean inequality. The equality holds if and only if \(r_1^2 + r_2^2 = 0\). Therefore, we obtain $1 - \frac{\sqrt{2}}{2} \left[ 2 - \frac{C_{l_1}^2(\rho)}{2} \right]^{\frac{1}{2}} \leq C_{\mathrm{diag}}^{F}(\rho)$ by Eq. \eqref{eq:diag_coherence} in Theorem \ref{th qubit eq}. 
        
    On the other hand, let the function \( g(t) = \sqrt{(1 - r_1^2 - r_2^2 - t)(1 - t)}  + (1 + t)\) with \( t \in [0, 1 - r_1^2 - r_2^2] \), its derivative 
      %  \(
    %    g'(t) = \frac{-(2 - r_1^2 - r_2^2 - 2t) + 2\sqrt{(1 - r_1^2 - r_2^2 - t)(1 - t)}}{2\sqrt{(1 - r_1^2 - r_2^2 - t)(1 - t)}}
%        \)
%  is non-positive, 
$g'(t) \leq 0$, which means
       \( g(t) \) is a monotonically decreasing function. Then it follows that $g(r_3^2) \geq g( 1 - r_1^2 - r_2^2)$, that is, \( \sqrt{(1 - \vert \boldsymbol{r} \vert^2)(1 - r_3^2)} + (1 + r_3^2) \geq 2 - r_1^2 - r_2^2 \).  We finally obtain \( C_{\mathrm{diag}}^{F}(\rho) \leq 1 - \frac{\sqrt{2}}{2} \left[ 2 - C_{l_1}^2(\rho) \right]^\frac{1}{2} \), and the equality holds when \( \vert \boldsymbol{r} \vert^2 = 1 \).
\end{proof}

Up to now we have already known that
$C_{\mathrm{diag}}^{F}$ is a well-defined coherence measure under the SIO framework. In fact, this quantity has been proposed  as a coherence monotone in \cite{Yadin2016}. But whether $C_{\mathrm{diag}}^{F}$ remains monotonic under IOs is still unknown. Here we provide a rigorous proof that $C_{\mathrm{diag}}^{F}$ is definitely monotonic under IOs in qubit systems. 

\begin{theorem}  \label{th IO}
In qubit systems, $C_{\mathrm{diag}}^{F}$ is monotonic under the IOs.
\end{theorem}

        \begin{proof}
    For any two qubit states \(\rho=(\rho_{ij})_{2\times 2}\) and \(\sigma=(\sigma_{ij})_{2\times 2}\). Using Eq. \eqref{eq qubit cdiag}, we derive that
        \begin{equation*}
            F(\rho, \Delta(\rho)) = \sqrt{1 - 2x_1 + 2\sqrt{x_1(x_1 - y_1)}}
        \end{equation*}
        and
        \begin{equation*}
            F(\sigma, \Delta(\sigma)) = \sqrt{ 1 - 2x_2 + 2\sqrt{x_2(x_2 - y_2)}}, 
        \end{equation*}
   with
        \( x_1 = \rho_{11}\rho_{22} \), \( y_1 = |\rho_{12}|^2 \), \( x_2 = \sigma_{11}\sigma_{22} \), and \( y_2 = |\sigma_{12}|^2 \). 
        By the positivity of the density matrix, we know
        \begin{equation}
             0 \leq x_1, x_2 \leq \frac{1}{4},\ \  x_1 \geq y_1,\ \   x_2 \geq y_2.
        \end{equation}
Suppose the quantum state $\rho$ can be transformed into $\sigma$ under IOs. Next we prove $C_{\mathrm{diag}}^{F}(\rho)\geq C_{\mathrm{diag}}^{F}(\sigma)$.

First, if $y_2=0$, then the quantum state $\sigma$ is incoherent and $C_{\mathrm{diag}}^{F}(\sigma)=0$. Therefore it is obvious that $C_{\mathrm{diag}}^{F}(\rho)\geq C_{\mathrm{diag}}^{F}(\sigma)$. 
%In this case, the coherence measure $C_{\mathrm{diag}}^{F}$ is obviously monotonic under IOs.

Second,  if $y_2> 0$, the quantum state $\rho$ can be transformed into $\sigma$ under IOs if and only if  \cite{Streltsov2017b,Cui2024}
\begin{equation}\label{eq cons}
  y_1 \geq y_2  \ \ \text {and}\ \  \frac{x_2}{x_1} \geq \frac{y_2}{y_1}.
\end{equation}
This implies that $y_1> 0$. By the positivity of density matrices, we have $x_1> 0$ and $x_2> 0$.
Let \(r_i = \frac{y_i}{x_i}\), \(i=1,2\). Then Eq.(\ref{eq cons}) can be rephrased as 
\begin{equation}\label{eq r}
r_1 \geq r_2 \ \  \text {and}\ \  \frac{r_2}{r_1} \leq \frac{x_1}{x_2}. 
\end{equation}
By the monotonicity of the function
$\eta(r_i) = \frac{1 - \sqrt{1 - r_i}}{r_i}$ on \((0, 1]\), it immediately follows that 
\begin{equation}\label{eq r2}
\frac{1 - \sqrt{1 -r_2}}{1 - \sqrt{1 -r_1}} \leq \frac{r_2}{r_1}.
\end{equation}
By Eq. (\ref{eq r}), this
 yields
\begin{equation*}
    x_1 (1 - \sqrt{1 -r_1})  \geq x_2(1 - \sqrt{1 -r_2}),
\end{equation*}
which is equivalent to 
\begin{equation*}
F(\rho, \Delta(\rho)) \leq F(\sigma, \Delta(\sigma)) . 
\end{equation*}
Thus we derive $C_{\mathrm{diag}}^{F}(\sigma) \leq C_{\mathrm{diag}}^{F}(\rho)$. Therefore $C_{\mathrm{diag}}^{F}$ is monotonic under the IOs in qubit systems.     
        \end{proof}

Since $C_{\mathrm{diag}}^{F}$ is nonnegative and additive under direct sum states, combined with its monotonicity under IOs, we know $C_{\mathrm{diag}}^{F}$ can serve as a coherence measure in the IO framework in qubit systems.

\begin{corollary}
    $C_{\mathrm{diag}}^{F}$ is a well-defined coherence measure in the IO framework for qubit states.
\end{corollary}  

%Compared with the coherence measures based on fidelity  \cite{Baumgratz2014,Streltsov2015,Liu2017,Li2019}, the advantage of the coherence measure $C_{\mathrm{diag}}^{F}$ we proposed here is the absence of optimization. 
%This makes the calculation and the application more feasible.
%Moreover, this coherence measure displays the essential role of the diagonal state $\Delta(\rho)$ in characterizing the coherence.
%The disadvantage of $C_{\mathrm{diag}}^{F}$ is that it is a reasonable coherence measure in the  IO framework only in qubit systems, and its monotonicity under the SIOs is still unknown in high dimensional systems.

\section{The application in multi-path interferometer} 
\label{interferometer}

In the multi-path interferometer, we denote \( |i\rangle \) as the \( i \)-th path, with \( i = 1, 2, \dots, d \).
Any quantum state $\rho$ in the interferometer can be expressed in the path basis $\{|i\rangle\}_{i=1}^{d}$ as  $\rho = \sum_{i,j=1}^{d} \langle i|\rho|j\rangle |i\rangle \langle j| $. 
To characterize the visibility of the interference fringes caused by the waveness,
the wave measure denoted by $\mathcal{W}(\rho)$ is introduced, which is a nonnegative function and satisfies the following conditions \cite{durr}:

   \begin{enumerate}[(W1)] 
   \item $\mathcal{W}(\rho)$ attains its global minimum
   when $\rho$ is diagonal in the path basis, i.e.,
   $\rho=\sum_{i=1}^{d} p_i |i\rangle\langle i|$. 
    \item $\mathcal{W}(\rho)$ attains its global maximum {when $\rho$ is pure state with equal diagonal entries i.e.  $\langle i|\rho|i\rangle = \frac{1}{d}$ for all $i$ }.
    \item $\mathcal{W}(\rho)$ is invariant under the permutations of the path labels.
    \item $\mathcal{W}(\rho)$ is convex.
\end{enumerate}
Dually, to characterize the path information in the interferometer, 
the particle measure denoted by $\mathcal{P}(\rho)$ is introduced which is a nonnegative function and satisfies the following conditions \cite{durr}:
\begin{enumerate}[(P1)] 
    \item  $\mathcal{P}(\rho)$ attains its global maximum when $\langle i|\rho|i\rangle = 1$ for some  $i$. 
    \item  $\mathcal{P}(\rho)$ attains its global minimum when $\langle i|\rho|i\rangle =  \frac{1}{d}$ for all $i$. 
    \item  $\mathcal{P}(\rho)$ is invariant under the permutations of path labels.
    \item  $\mathcal{P}(\rho)$ is convex.
\end{enumerate}

In fact, the mutual exclusivity between waveness and particleness is reflected by the requirements of the wave measure and particle measure. When the waveness attains the maximum,  the particleness reaches the minimum. This is the case that the state is equally likely to occur across all possible paths. Conversely, when the particleness attains the maximum, then the waveness reaches the minimum. This is the case that a path is traversed with certainty and the systems show no interference.

By the analysis of elementary process in the interferometer \cite{Yadin2016}, we know the operations acting on a particle travelling through an interferometer are all SIOs. Indeed, it can be verified that the coherence measures in the SIO framework all satisfy the conditions required for the wave measure (W1)-(W4). 
Therefore, to quantify the waveness, we can adopt the coherence measure in the SIO framework which extends the method proposed in \cite{Bai2025}.
In particular, 
\begin{equation}
    W(\rho)=C_{\mathrm{diag}}^{F}(\rho),
\end{equation}
can be viewed as the wave measure. Correspondingly, 
we define 
\begin{equation}\label{eq P}
    P(\rho) = 1 - F(\Delta(\rho), \frac{I}{d}),
\end{equation}
as the particle measure.
Indeed, for any quantum state $\rho=\sum_{i,j} \rho_{ij} |i\rangle\langle j|$, 
$P(\rho)$ in Eq. (\ref{eq P}) can be rewritten as 
\begin{equation}
    P(\rho) = 1 - \frac{1}{\sqrt{d}}\sum_{i=1}^d \sqrt{\rho_{ii}}.
\end{equation}
It can be verified that $0\leq P(\rho) \leq 1 -  \frac{1}{\sqrt{d}}.$
The minimum is attained, $P(\rho)=0$, if and only if $\rho_{ii} = \frac{1}{d}$ for all $i$. The maximum is attained, $P(\rho) = 1 -  \frac{1}{\sqrt{d}}$, if and only if $\rho_{ii} = 1$ for some $i$. 
Moreover, \( P(\rho) \) is invariant under permutations of the path labels, and its convexity is guaranteed by the concavity of the fidelity \( F(\rho) \). Therefore, $P(\rho) $ characterizes the particle behavior of $\rho$ in the multi-path interferometer. 

\begin{theorem}\label{th 2 duality}
For any quantum state $\rho$ in two-path interferometer, the wave  measure $W(\rho)$ and the particle measure $P(\rho)$ satisfy 
\begin{equation}\label{eq dualtiy}
    W(\rho) + P(\rho) \leq 1-\frac{1}{\sqrt{2}}.
    \end{equation}
The equality holds if and only if $\rho$ is classical, i.e., $|1\rangle\langle 1|$, $|2\rangle\langle 2|$, or is maximally coherent state $\frac{1}{\sqrt{2}}(|1\rangle + e^{{\rm i}\theta}|2\rangle)$ with arbitrary phase $\theta$.
\end{theorem}

\begin{proof}
Since both $W(\rho)$ and  $P(\rho)$ are convex, so in order to find the maximum of the sum $W(\rho)+ P(\rho)$, we only need to consider the case of pure states. For any pure qubit state \( |\psi\rangle = a|1\rangle + b|2\rangle \), where \( |a|^2 + |b|^2 = 1 \). 
%Let \( x = |a| \), \( y = |b| \), 
By definition and direct calculation, we get that its particleness is  
\[
P(\rho) = 1 - \frac{|a| + |b|}{\sqrt{2}},
\]
and waveness is
\[
W(\rho) = 1 - \sqrt{|a|^4 + |b|^4},
\]
respectively.
Therefore the sum is
\begin{eqnarray*}
W(\rho) + P(\rho) &=& 2 - \sqrt{|a|^4 + |b|^4} - \frac{|a| + |b|}{\sqrt{2}}\\
&=& 2- \sqrt{1 - 2|ab|^2} -\sqrt{\frac{1}{2} + |ab|}.
\end{eqnarray*}
Let \( t = |ab| \) (\( 0 \leq t \leq \frac{1}{2} \)), and let
\[
f(t) = \sqrt{1 - 2t^2} + \sqrt{\frac{1}{2} + t}.
\]
Now we analyze the minimum of $f(t)$ at the interval $[0, \frac{1}{2}]$.
Because $f(t)$ is strictly concave about $t$ at the interval $[0, \frac{1}{2}]$, its minimum is attained at the  extreme points. Examing   these points, we get
\[
f(0) = f(\frac{1}{2}) =  1 +\frac{1}{\sqrt{2}}.
\]
Hence, for all \( t \in [0, \frac{1}{2}] \), \( f(t) \geq 1 +\frac{1}{\sqrt{2}}\), with equality when \( t = 0 \) (i.e., \( a = 0 \) or \( b = 0 \)) or \( t = \frac{1}{2} \) (i.e., \( a = b = \frac{1}{\sqrt{2}} \)).  Therefore, 
\begin{eqnarray*}
W(\rho) + P(\rho)\leq 1- \frac{1}{\sqrt{2}}
\end{eqnarray*}
and the equality holds true if and only if \(\rho\) is classical, i.e., \(\rho = |i\rangle \langle i|\) $(i=1,2)$ or maximally coherent $|\psi\rangle=\frac{1}{\sqrt{2}}(|1\rangle+e^{{\rm i}\theta}|2\rangle)$.
\end{proof}

Theorem \ref{th 2 duality} shows a tight wave-particle duality in two-path interferometer. In the multi-path interferometer, we show the waveness and particleness are restricted by the following bounds.

\begin{theorem}\label{th duality-1}
For any quantum state $\rho$ in $d$-path interferometer, the wave  measure $W(\rho)$ and the particle measure $P(\rho)$ satisfy 
\begin{equation}\label{eq dualtiy-1}
    W(\rho) + P(\rho) \leq 2 - \frac{2}{\sqrt[4]{d}}.
    \end{equation}
\end{theorem}

%\begin{proof}For any quantum state $\rho=\sum_{i,j} \rho_{ij} |i\rangle \langle j|$ in $d$-path interferometer, according to the definition of fidelity and Theorem \ref{properties}, we get
%\begin{eqnarray*}F(\rho, \Delta(\rho)) + F\left(\Delta(\rho), \frac{I}{d}\right) &\geq& \sqrt{\displaystyle\sum_{i}\rho_{ii}^2} + \frac{1}{\sqrt{d}} \displaystyle\sum_{i}\sqrt{\rho_{ii}}\\ &\geq& \displaystyle\sum_{i} \rho_{ii} \sqrt{\rho_{ii}} + \frac{1}{\sqrt{d}} \displaystyle\sum_{i}\sqrt{\rho_{ii}}\\&\geq& \frac{2}{\sqrt[4]{d}}\displaystyle\sum_{i} \rho_{ii} \\&=& \frac{2}{\sqrt[4]{d}}, 
%\end{eqnarray*}where we have used the concavity for the second inequality and mean inequality $a + b \geq 2\sqrt{ab}$ for the third inequality. This makes the wave-particle duality $ W(\rho) + P(\rho) \leq 2 - \frac{2}{\sqrt[4]{d}}$.\end{proof}
\begin{proof}
For any quantum state $\rho=\sum_{i,j} \rho_{ij} |i\rangle \langle j|$ in $d$-path interferometer, according to the definition of fidelity and Theorem \ref{properties}, we get
\begin{eqnarray*}
    F(\rho, \Delta(\rho)) + F\left(\Delta(\rho), \frac{I}{d}\right) &\geq& \sqrt{\displaystyle\sum_{i}\rho_{ii}^2} + \frac{1}{\sqrt{d}} \displaystyle\sum_{i}\sqrt{\rho_{ii}}\\
    &\geq& \displaystyle\sum_{i} \rho_{ii} \sqrt{\rho_{ii}} + \frac{1}{\sqrt{d}} \displaystyle\sum_{i}\sqrt{\rho_{ii}}\\
     &\geq& \frac{2}{\sqrt[4]{d}}\displaystyle\sum_{i} \rho_{ii} \\
     &=& \frac{2}{\sqrt[4]{d}}, 
\end{eqnarray*}
where we have used the fact that \(\sqrt{\sum_i \rho_{ii}^2} \geq \sum_i \rho_{ii}^{\frac{3}{2}}\) by the concavity of the square-root function
for the second inequality. For the third inequality, we have used the fact that $ \rho_{ii}^{\frac{3}{2}} + \frac{1}{\sqrt{d}} \rho_{ii}^{\frac{1}{2}} \geq \frac{2}{\sqrt[4]{d}} \rho_{ii}$ by the arithmetic-geometric mean (AM-GM) inequality
$a+b\geq 2\sqrt{ab}$. 
This makes the wave-particle duality $ W(\rho) + P(\rho) \leq 2 - \frac{2}{\sqrt[4]{d}}$.
\end{proof}

Furthermore, the upper bound in Theorem \ref{th duality-1} can be improved in a finer way.

\begin{theorem}\label{th duality}
For any quantum state $\rho$ in $d$-path interferometer with $\Delta(\rho)=\sum_i \langle i|\rho |i\rangle |i\rangle \langle i|$ of rank $k$ $(k\leq d)$ under the path basis, the wave  measure $W(\rho)$ and the particle measure $P(\rho)$ satisfy 
\begin{equation}\label{eq dualtiy}
    W(\rho) + P(\rho) \leq 2 - \frac{2\sqrt{2}}{3} \sqrt{\frac{3}{\sqrt{d}} + \frac{k}{d}}.
    \end{equation}
\end{theorem}

\begin{proof}
For any quantum state $\rho=\sum_{i,j} \rho_{ij} |i\rangle \langle j|$ in $d$-path interferometer, by the proof in Theorem \ref{th duality-1}, we know
\begin{eqnarray*}
    F(\rho, \Delta(\rho)) + F\left(\Delta(\rho), \frac{I}{d}\right) 
    \geq \displaystyle\sum_{i} \rho_{ii} \sqrt{\rho_{ii}} + \frac{1}{\sqrt{d}} \displaystyle\sum_{i}\sqrt{\rho_{ii}}.
\end{eqnarray*}
Since $\Delta(\rho)$ is of rank $k$,
without loss of generality, suppose $\rho_{ii}> 0$ for $i=1,2,\cdots,k$ and  $\rho_{ii}=0$ for $i=k+1,\cdots,d$, with $1\leq k\leq d$. 
Next we consider the minimum of the function
\begin{equation}\label{eq f}
    f(\rho_{11}, \rho_{22}, \cdots, \rho_{kk})=\sum_i \rho_{ii}\sqrt{\rho_{ii}}+\frac{1}{\sqrt{d}}\sum_i \sqrt{\rho_{ii}},
\end{equation}
with the constraints $\left\{(\rho_{11}, \dots , \rho_{kk})\in \mathbb{R}^k: \rho_{ii}> 0, \sum_{i=1}^k \rho_{ii}=1\right\}$. 
We construct its Lagrangian function
\[
\mathcal{L}(\rho_{ii}, \lambda) = \sum_i \left( \rho_{ii}^{\frac{3}{2}} + \frac{1}{\sqrt{d}} \rho_{ii}^{\frac{1}{2}} \right) + \lambda \left( \sum_i \rho_{ii} - 1 \right).
\]
Then, by derivative we get 
\begin{equation}\label{l}
\frac{\partial \mathcal{L}}{\partial \rho_{ii}} = \frac{3}{2} \rho_{ii}^{\frac{1}{2}} + \frac{1}{2\sqrt{d}} \rho_{ii}^{-\frac{1}{2}} + \lambda = 0,\ \forall i,
\end{equation}
which implies $\lambda < 0$. Multiplying  Eq. (\ref{l}) by $2\sqrt{\rho_{ii}}$ and summing over all $i$, and then applying $\sum_{i}\rho_{ii}=1$, we obtain
%\begin{equation}
 %   3+\sqrt{d}+2\lambda \sum_i \sqrt{\rho_{ii}}=0.
%\end{equation}
%This equation gives rise to 
\begin{equation}\label{eq sqrt rho}
     \sum_i \sqrt{\rho_{ii}}=  -\frac{3 + \frac{k}{\sqrt{d}}}{2\lambda}.
\end{equation}
Multiplying  Eq. (\ref{l}) by $2{\rho_{ii}}$ and summing over all $i$, and then applying Eq. (\ref{eq sqrt rho}), we obtain
\begin{equation}\label{eq sqrt rho11}
     \sum_{i} \rho_{ii}^{\frac{3}{2}} = \frac{3 + \frac{k}{\sqrt{d}}}{6\lambda\sqrt{d}} - \frac{2\lambda}{3}.
\end{equation}
Inserting Eqs. (\ref{eq sqrt rho}) and (\ref{eq sqrt rho11}) into the function $f$, we derive the function $f$ as the function of $\lambda$ as
\begin{equation}\label{eq f 2}
f(\lambda) = -\frac{3 + \frac{k}{\sqrt{d}}}{3\lambda\sqrt{d}} - \frac{2\lambda}{3}.
\end{equation}
Now we consider the minimum of the function $f$ in Eq. (\ref{eq f 2}) in the interval $(-\infty, 0)$.
Taking the derivative of $f$ in Eq. (\ref{eq f 2}), we get
\[
f'(\lambda) = \frac{3 + \frac{k}{\sqrt{d}}}{3\sqrt{d}{\lambda^2}} - \frac{2}{3},
\quad
f''(\lambda) = -\frac{2\left(3 + \frac{k}{\sqrt{d}}\right)}{3\sqrt{d}{\lambda^3}}.
\]
Since $\lambda < 0$, which implies $f''(\lambda) > 0$. Thus, $f(\lambda)$ is a convex function for $\lambda < 0$.
Setting $f'(\lambda) = 0$, we solve for the unique stationary point $\lambda_0 =-\sqrt{\frac{3 + \frac{k}{\sqrt{d}}}{2\sqrt{d}}} < 0.$
By convexity, $\lambda_0$ is the global minimum of $f(\lambda)$. The boundary behavior of the function is given by
$\lim_{\lambda\to 0^-} f(\lambda) = +\infty$ and $\lim_{\lambda\to -\infty} f(\lambda) = +\infty$.
Therefore the minimum value of $f$ in the interval $(-\infty, 0)$ is 
\begin{equation}
\label{fmin}
f_{\min} = \frac{2\sqrt{2}}{3} \sqrt{\frac{3}{\sqrt{d}} + \frac{k}{d}}.
\end{equation}
This demonstrates that the minimum of $f$ in Eq. (\ref{eq f}) is $f_{\min} $.
\end{proof}

Theorem \ref{th duality} provides the wave-particle duality depending on the path number $d$ as well as the rank $k$ of $\Delta(\rho)$.
This is the situation that  we only focus on the $k$ paths in the $d$-path interferometer. It is analogous with  a post-selecting process on the interfering port. 
For example, in the Franson interferometer, there are four possible
paths but post-selecting on coincidence counts discards two of these paths, which is factually a binary interferometer  \cite{Coles}.
%Similarly, although there are $d$ paths in the interferometer, the number of actually effective paths is $k$, where $k \le d$. In this case, $\Delta(\rho)$  considered is of rank $k$. 
Moreover, this upper bound in Eq. (\ref{eq dualtiy}) is strictly better than that in Theorem \ref{th duality-1} when $k > \frac{3}{2}\sqrt{d}$. Therefore, 
if we know the rank of $\Delta(\rho)$ for the input quantum state, then the sum of the waveness and particleness can be estimated better.
In particular, if $\Delta(\rho)$ is full rank, that is $k=d$, then the  wave-particle duality in Theorem \ref{th duality} turns into
\begin{equation}\label{eq full delta rho}
W(\rho) + P(\rho) \leq 2 - \frac{2\sqrt{2}}{3} \sqrt{\frac{3}{\sqrt{d}} + 1}. 
\end{equation}
In this case, the upper bound in Eq. (\ref{eq full delta rho}) is strictly better than that in Eq. (\ref{eq dualtiy-1})
for $d\geq 3$.
Furthermore, when the path number $d$ tends to infinity,
the upper bound in Eq. (\ref{eq full delta rho}) tends to $2-\frac{2\sqrt{2}}{3}\approx1.05719$. This demonstrates the upper bound gets more precise as $d$ gets larger.

It is intriguing to explore the tight wave-particle duality in the multi-path interferometer. As both $ W(\rho)$ and  $P(\rho)$ are convex, now we consider the maximum of $ W(\rho)+ P(\rho)$ for pure states.
For the pure states $|\psi\rangle=\frac{1}{\sqrt{d}} \sum_{i=1}^d  |i\rangle$ and $|\psi\rangle=|0\rangle$, by direct calculation, we find that they both satisfy $ W(|\psi\rangle) + P(|\psi\rangle)=1-\frac{1}{\sqrt{d}}.$
One intuitive upper bound for the sum of waveness and particleness might be 
\begin{eqnarray}\label{eq linear wp}
    W(|\psi\rangle) + P(|\psi\rangle)\leq 1-\frac{1}{\sqrt{d}}.
\end{eqnarray} 
In fact, it is definitely true for $d=2$ as proved in Theorem \ref{th 2 duality}. However, for $d=3$, 
we have generated  \( n = 10,000 \) randomly sampled pure quantum states and calculated their waveness and particleness. 
Our numerical results demonstrate that there exist some pure states violating the bound, that is, \(W(\rho) + P(\rho) > 1 - \frac{1}{\sqrt{3}}\) (See FIG. \ref{fig:figure1}).

\begin{figure}[htbp] 
    \centering 
    \includegraphics[width=0.4\textwidth]{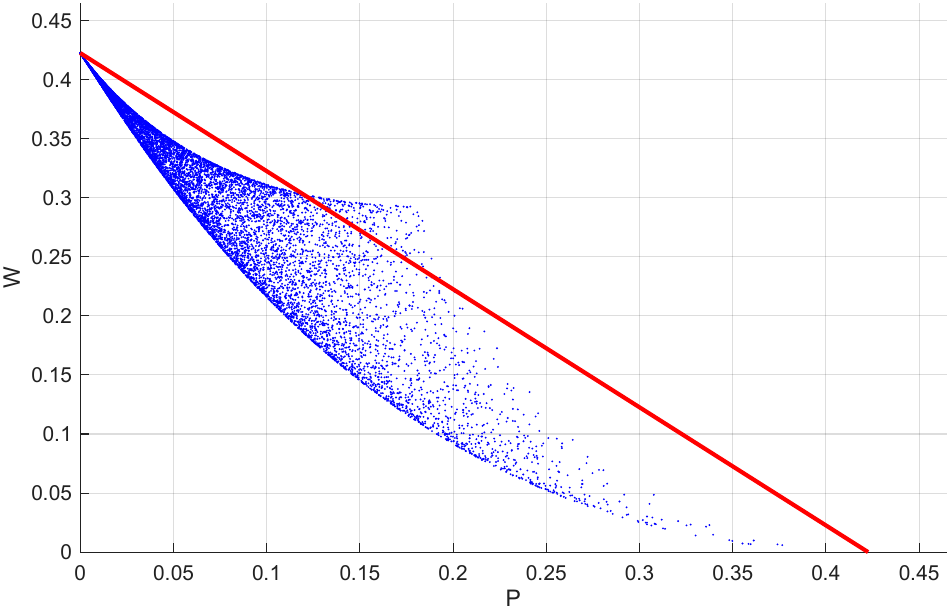} 
    \caption{The $P-W$ diagrams for the randomly sampled pure states in three-path interferometer. The red line represents \( W + P = 1 - \frac{1}{\sqrt{3}} \). The blue points correspond to the coordinates $(P, W)$ for \( n = 10,000 \) randomly sampled pure quantum states.}
    \label{fig:figure1}
\end{figure}

Since the linear form in Eq. (\ref{eq linear wp}) is imprecise for describing the constraint on waveness and particleness for high dimensional systems, we make an attempt to the quadric form
\begin{eqnarray}\label{eq qua form}
    W^2+P^2\leq 1-\frac{1}{\sqrt{d}}.
\end{eqnarray}
As shown in Fig. \ref{fig:figure3}, for 5000 randomly generated pure states in quantum systems with dimensions $d = 3, 10, 25, 50$, all the sampled points lie inside the region bounded by $W^2+P^2=1-\frac{1}{\sqrt{d}}$. This shows that the quadric form in Eq. (\ref{eq qua form}) probably gives a description of the wave-particle duality in high dimensional systems.

\begin{figure}[htbp] 
    \centering 
    \includegraphics[width=0.8\textwidth]{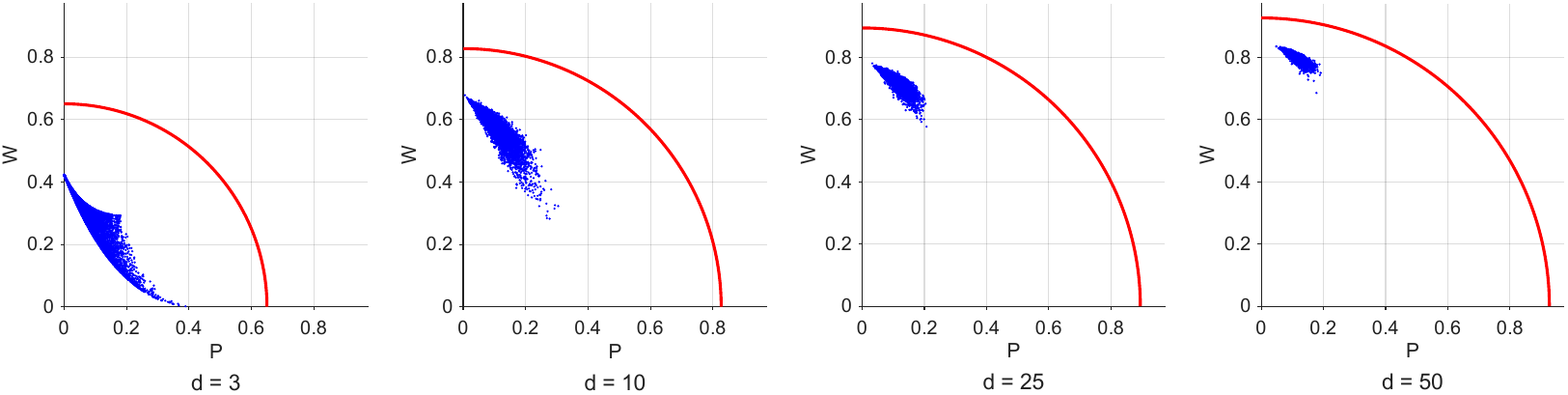} 
    \caption{The $P-W$ diagrams for randomly sampled pure quantum states in dimensions
    $d = 3, 10, 25, 50$.
    The red curve represents the boundary $W^2 + P^2 = 1 - \frac{1}{\sqrt{d}}$,
    and the blue points denote the coordinates $(P, W)$ for $n = 5000$
    randomly generated pure states.}
    \label{fig:figure3}
\end{figure}

\section{Conclusions}
\label{Conclusions}

To summarize, we have studied the coherence measures in the SIO framework. We
have presented a method to construct coherence measures using the diagonal state $\Delta(\rho)$. This class of coherence measures includes some existing coherence measures in the form of the diagonal state. Specifically, we have proposed a coherence measure $C_{\mathrm{diag}}^{F}$ in terms of fidelity and presented its analytical expressions. 
The relationships between the proposed coherence measure $C_{\mathrm{diag}}^{F}$ and some other coherence measures are established.
Furthermore, we prove the monotonicity of $C_{\mathrm{diag}}^{F}$ under any IOs in qubit systems. So $C_{\mathrm{diag}}^{F}$ is indeed a coherence measure in the IO framework in qubit systems. 
Last but not least, we discuss the role of the coherence measures in the SIO framework in characterizing the waveness in the multi-path interferometer. We derive several wave-particle dualities in terms of fidelity.
%Overall, this work not only underlines the difference between the coherence measures in the IO and SIO frameworks, but also further breaks the equivalence between the coherence in the IO framework and waveness in the multi-path interferometer.

It is interesting to consider the extension of the coherence
measure $C_{{\rm diag}}^{F}$ to the coherence resource theory \cite{JWXu2016} in infinite-dimensional
(continuous-variable) systems, particularly the (bosonic) Gaussian states, and discuss the corresponding wave-particle duality in the interferometric setups.
However, due to the divergence and computational issues in such systems, this problem is challenging and constitutes a promising direction for future research.

\section{Acknowledgement}
We would like to thank Jianwei Xu and Yue Sun for valuable discussions.
We thank also the anonymous referees for helpful
comments.
M. J. Zhao thanks the center for Quantum Information,
Institute for Interdisciplinary Information Sciences,
Tsinghua University for hospitality.
This work is supported by the National Natural Science Foundation of China under grant No. 12171044.

\end{document}